\documentclass[10pt,conference]{IEEEtran}
\usepackage{xcolor,soul,framed} 
\usepackage[noadjust]{cite}

\colorlet{shadecolor}{yellow}

\usepackage[pdftex]{graphicx}
\graphicspath{{../pdf/}{../jpeg/}}
\DeclareGraphicsExtensions{.pdf,.jpeg,.png}
\usepackage[caption=false, font=footnotesize]{subfig}
\usepackage{amssymb,mathtools}
\usepackage{array}
\usepackage{mdwmath}
\usepackage{mdwtab}
\usepackage{eqparbox}
\usepackage{url}
\usepackage[ruled,vlined,linesnumbered]{algorithm2e}
\usepackage{enumitem} 
\usepackage[mathcal]{euscript} 
\usepackage{tabularx} 
\usepackage{multirow} 
\usepackage{booktabs}
\usepackage{makecell}
\usepackage{aligned-overset}
\usepackage{url}
\usepackage{multicol}

\usepackage{balance}
\usepackage{bbm}
\usepackage{amsmath}
\usepackage{amsthm} 
\usepackage{pifont}

\graphicspath{{./Figures/}}

\newtheorem{Theorem}{Theorem}

\renewcommand{\vec}[1]{\boldsymbol{\mathrm{#1}}}

\newcommand{\IRS}{\mathtt{IRS}}

\newcommand{\user}{\mathtt{U}}

\newcommand{\ap}{\mathtt{AP}}

\newcommand{\pl}{\ell}

\usepackage[hidelinks]{hyperref}
\usepackage{xcolor}
\usepackage[margin=0.8in]{geometry}

\begin{document}
\title{Dual-Tier IRS-Assisted Mid-Band 6G Mobile Networks: Robust Beamforming and User Association}
\author{
  \IEEEauthorblockN{Muddasir Rahim and Soumaya Cherkaoui}
  \IEEEauthorblockA{Department of Computer Engineering and Software Engineering, Polytechnique Montreal, Canada} \IEEEauthorblockA{Emails: muddasir.rahim@polymtl.ca, soumaya.cherkaoui@polymtl.ca}
}
\maketitle

\begin{abstract}
The rapid growth of Internet of Things (IoT) applications necessitates robust resource allocation in future sixth-generation (6G) networks, particularly at the upper mid-band (7–15 GHz, FR3). This paper presents a novel intelligent reconfigurable surface (IRS)-assisted framework combining terrestrial IRS (TIRS) and aerial IRS (AIRS) mounted on low-altitude platform stations, to ensure reliable connectivity under severe line-of-sight (LoS) blockages. Distinguishing itself from prior work restricted to terrestrial IRS and mmWave/THz bands, this work targets the FR3 spectrum, the so-called "Golden Band" for 6G. The joint beamforming and user association (JBUA) problem is formulated as a mixed-integer nonlinear program (MINLP), solved through problem decomposition, zero-forcing beamforming, and a stable matching algorithm. Comprehensive simulations show our method approaches exhaustive search performance with significantly lower complexity, outperforming existing greedy and random baselines. These results provide a scalable blueprint for real-world 6G deployments, supporting massive IoT connectivity in challenging environments.
\end{abstract}
\section{INTRODUCTION}\label{introduction}
\IEEEPARstart{T}{he} rapid expansion of Internet of Things (IoT) applications has introduced diverse and demanding requirements for wireless communication networks, such as ultra-high data rates, low latency, and high reliability. To meet these growing needs and support emerging data-intensive applications such as virtual/augmented reality (VR/AR) and ultra-high-definition video streaming, sixth-generation (6G) networks must provide more capacity. Furthermore, for commercial viability, 6G must also accommodate a wide range of new use cases and functionalities while ensuring reliable connectivity for a massive number of heterogeneous devices. 

To meet these requirements, one solution is to use higher-frequency bands with wider bandwidths, thereby enabling higher data rates. Over the past decade, the millimeter-wave (mmWave) band (24 GHz and above) was considered the natural next step for wireless communication due to its abundant spectrum availability. As a result, 5G was co-designed to operate over two frequency ranges (FRs): FR1, spanning 0.4–7 GHz, and FR2, covering 24–71 GHz~\cite{bjornson2025enabling}. However, the high susceptibility of mmWave bands to blockage results in intermittent coverage, posing a significant challenge that necessitates costly dense small-cell deployments to ensure reliable connectivity. Although the research community is actively exploring even higher-frequency bands, such as the terahertz (THz) and sub-THz bands, these are expected to inherit and amplify the challenges associated with mmWave frequencies~\cite{rappaport2019wireless}. Therefore, there is growing interest from both academia and industry in identifying suitable frequency bands, with particular attention to the upper mid-band (UMB) (7–15 GHz), also known as FR3~\cite{nokia2024}. Several key factors are driving the selection of FR3. Firstly, UMB offers significantly larger bandwidth than FR1. Moreover, it exhibits more favorable propagation characteristics than FR2, resulting in substantially better coverage. Therefore, the FR3 is now referred to as the "Golden Band" for 6G, owing to its promising balance between bandwidth availability and coverage~\cite{cui20236g}. 
Yet translating FR3’s spectral advantages into reliable coverage at scale requires addressing non-negligible propagation losses and frequent LoS blockages in cluttered environments; constraints that motivate environment-shaping solutions beyond densified small cells.

To overcome propagation losses and enhance coverage, several solutions have been proposed, including ultra-massive multi-input-multiple-output (UM-MIMO) and intelligent reflecting surfaces (IRSs). However, UM-MIMO systems typically incur high power consumption and hardware complexity. The large number of antenna elements and associated radio frequency (RF) chains significantly increase energy usage, signal processing demands, and implementation costs, making their large-scale deployment challenging. In contrast, IRSs provide a low-cost, energy-efficient alternative by enabling reconfigurable radio environments through large arrays of passive reflecting elements~\cite{huang2019reconfigurable}. Most existing IRS-assisted networks primarily focus on terrestrial deployments, where IRSs are mounted on buildings to enhance coverage. However, in dense urban environments where line-of-sight (LoS) links are frequently obstructed, relying solely on terrestrial IRSs (TIRSs) can lead to coverage blind spots~\cite{rahim2023joint,rahim2024multi}. To extend coverage and eliminate these blind spots, a promising solution is to integrate low-altitude platform station (LAPS)-mounted IRSs, also known as aerial IRSs (AIRSs), into the network~\cite{you2022enabling}. These AIRSs offer distinct advantages in terms of coverage extension, service continuity, and rapid deployment. Nevertheless, the introduction of such a multi-tier, heterogeneous architecture brings significant challenges in resource allocation (RA). Thus, effective RA is critical in 6G networks. The FR3 remains unexplored, mainly in the context of IRSs, therefore, IRS-assisted FR3 band communication represents a vital research direction with significant potential to impact real-world 6G network deployments.
 
\subsection{Related Works}
The 6G standardization begins in 2025 and
will be completed in 2028-2029, allowing commercial 6G networks to open in 2029-2030. These networks will likely operate in FR3. In this context, many researchers have started to study this band, including its achievable communication performance and applications. Meanwhile, only a few papers in the literature deal with MIMO in FR3~\cite{heath2024beamsharing,bjornson2025enabling,tian2024mid}. In~\cite{heath2024beamsharing}, the authors proposed a user-pairing method and used beam sharing for near-field and far-field users in a MIMO setup. Moreover, in~\cite{bjornson2025enabling}, the authors introduced a term gigantic MIMO (gMIMO) (i.e., at least 256 antenna ports), indicating an even larger number of antennas than in UM-MIMO. This work explored the potential of gMIMO in FR3. The authors in~\cite{tian2024mid} integrated extra large-scale MIMO (XL-MIMO) in mid-band communications. The aim of this paper was to propose an analytical model and analyze key metrics of wireless networks, including spectral efficiency (SE) and outage probability (OP).

In~\cite{mohsan2023irs}, the authors presented a comprehensive survey of IRS-aided unmanned aerial vehicle (UAV) communications, summarizing architectures, channel considerations, and design objectives. Authors in~\cite{abdalla2020uavs} focused on the promise of integrating IRS with UAV platforms for future cellular systems, outlining multiple use cases, key challenges, and open research directions. They discussed related works in the domain of spectrum sharing, physical layer security, giant-site access, and enhanced coverage. Moreover, in~\cite{bui2025joint}, the authors proposed a joint optimization framework for transmit power allocation and IRS phase-shift configuration. This framework aimed to maximize energy efficiency while maintaining reliable communication for all users. IRSs were deployed on buildings to enhance the coverage and provide a reflective link between the UAVs and users. Meanwhile, the authors in~\cite{farre2025dynamic} proposed integrating non-terrestrial networks (NTNs) with IRSs into existing terrestrial networks (TNs). This work aimed to enhance connectivity and capacity in crowded environments by optimizing coverage and interference mitigation. IRS-assisted UAV communication and AIRS-assisted terrestrial communication are two new schemes proposed by the authors of~\cite{you2022enabling} to jointly utilize IRS and UAV for next-generation wireless networks with integrated terrestrial and aerial communications.

The FR3 bands are the least explored in the IRS context. To the best of the authors' knowledge, the only article in the literature that demonstrates the significance of IRSs in the FR3 is~\cite{kara2024reconfigurable}. This work explored different scenarios under which the IRS can provide significant benefits and optimal strategies for deploying IRSs to enhance the performance of the FR3 communication system. Therefore, realizing the full potential of integrating AIRS and TIRS in the FR3 for 6G requires further research. In particular, IRS studies overwhelmingly emphasize FR2 (mmWave/THz) and single-tier terrestrial deployments, leaving multi-tier FR3 operation comparatively underexplored, despite FR3’s emerging role in 6G and its distinct blockage and coverage trade-offs.
\subsection{Contributions}\label{contributions}
The aforementioned studies have primarily examined IRS-assisted communication in terrestrial networks and have largely focused on mmWave/THz bands. This leaves a gap regarding FR3 operation and multi-tier architectures. Motivated by this gap, we investigate the integration of TIRSs with AIRSs for FR3 communications. The following are the main contributions of this paper:

\begin{itemize}
    \item To the best of our knowledge, this is the first study that analyzes the impact of TIRSs and AIRS in FR3 bands. We formulate a joint beamforming and user association (JBUA) problem to maximize the downlink sum rate.
    \item The JBUA problem is a mixed-integer nonlinear program (MINLP) and is challenging to solve directly. Therefore, we decompose it into two subproblems: beamforming matrix optimization and user association. Then, solve each subproblem efficiently to construct a feasible suboptimal solution to the original problem.
    \item We employ zero-forcing (ZF) beamforming at the access point (AP) and a user-proposing deferred-acceptance matching algorithm for the association sub-problem, providing a zero-interference precoder and a stable one-to-one user–IRS matching mechanism.
    \item The proposed JBUA achieves sum-rate performance approaching that of exhaustive search (ES) while clearly outperforming the random search (RS) and greedy search (GS) schemes.
 
\end{itemize}

\section{System Model} \label{model}
A multi-layer IRS-assisted network consisting of a single AP equipped with $N$ antennas serving $K$ single-antenna users is shown in Fig.~\ref{m1}. This layout explicitly targets FR3 propagation with frequent blockages, motivating hierarchical placement to sustain controllable reflected paths.
To support reliable downlink communications in challenging environments, a total of $L$ IRSs are deployed, consisting of TIRSs and AIRSs. The corresponding IRS set is denoted as $\mathcal{L} = \{\IRS_1, \IRS_2, \ldots, \IRS_l, \ldots, \IRS_L\}$. Each IRS is equipped with $M = M_y M_z$ reflecting elements arranged in a uniform planar array, where $M_y$ and $M_z$ represent the number of elements along the horizontal and vertical axes, respectively. Due to severe signal blockages in the FR3 bands, all direct AP-to-user links are assumed to be unavailable. This assumption is reasonable when obstacles in a dense urban environment severely impede the direct link.
\begin{figure}[!t]
     \centering
\includegraphics[width=0.88\linewidth]{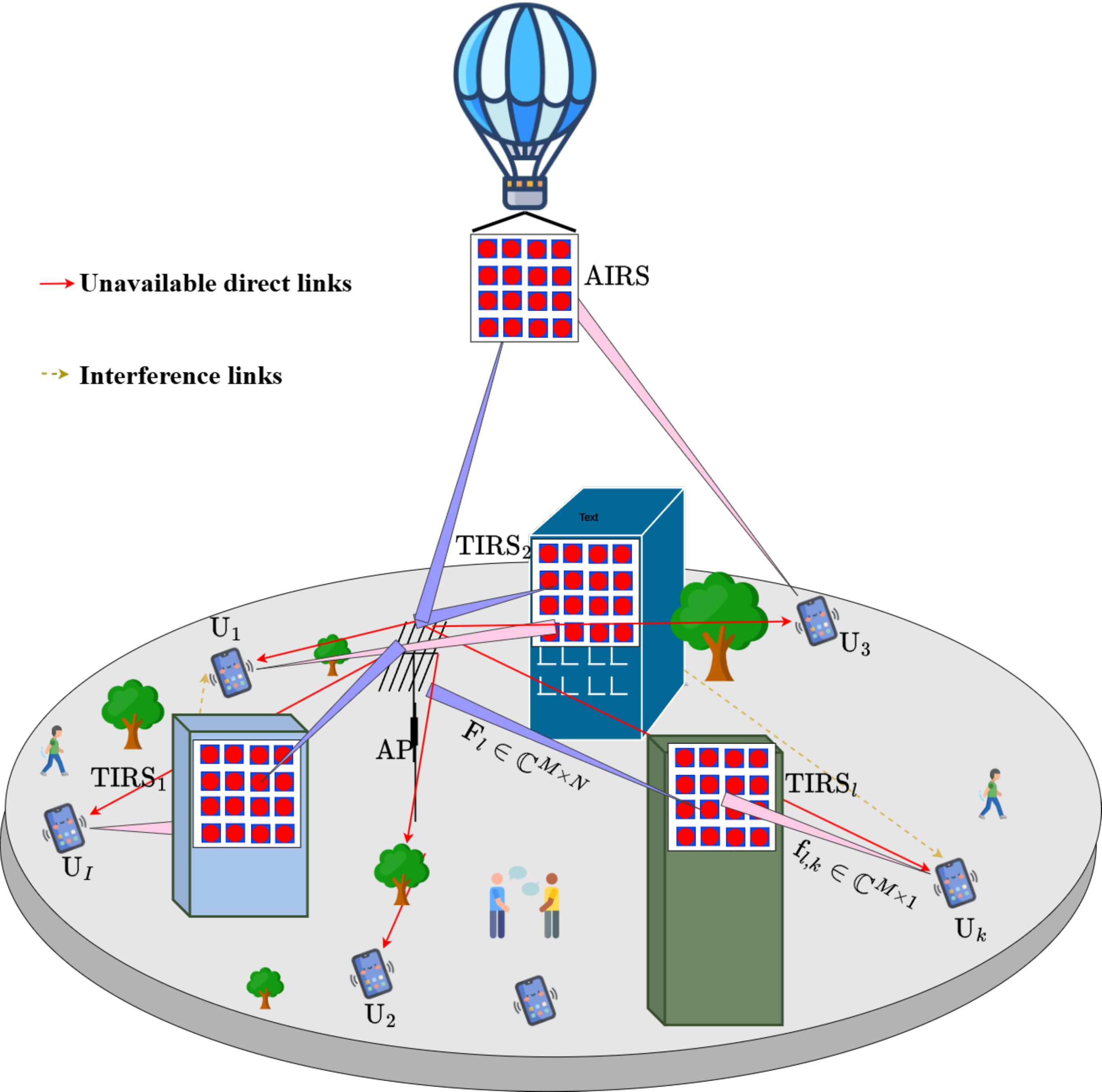}
    \caption{System model of IRS-assisted with LAPS and terrestrial layers.
     }
     \label{m1}   
\end{figure}
In the Cartesian coordinate system, the center locations of the AP, $k^{th}$ user ($\user_{k}$), and $l^{th}$ IRS ($\IRS_{l}$) are represented as ${\mathbf p_{a} = (x_{a}, y_a, z_{a})}$, ${\mathbf p_{k} = (x_{k}, y_{k}, z_{k})}$, and ${\mathbf p_{l} = (x_{l}, y_{l}, z_{l})}$, respectively. Furthermore, the distance from the AP to the center of $\IRS_{l}$ is written as 
\begin{align}
    d_{a, l} = 
    \sqrt{(x_{a}-x_{l})^2+(y_{a}-y_{l})^2+(z_{a}-z_{l})^2}.
\end{align}
Similarly, the distance from $\IRS_{l}$ to $\user_{k}$ is expressed as
\begin{align}
     d_{l, k} = 
    \sqrt{(x_{l}-x_{k})^2+(y_{l}-y_{k})^2+(z_{l}-z_{k})^2}.
\end{align}
\subsection{Channel Modeling and Data Rate Analysis} 
Considering $\vec{F}_l \in \mathbb{C}^{M \times N}$ denotes the channel matrix from the AP to ($\IRS_l$). The channel from the $n^{th}$ antenna of the AP ($\ap_n$) to the $m^{th}$ element of the $\IRS_l$ ($\IRS_{l,m}$) is then given as
 \begin{align}\label{eq1}
    f_{n,(l,m)} = \sqrt{\pl_{n,(l,m)}} \, e^{-j \omega d_{a,l}},
 \end{align}
where $\pl_{n,(l,m)}$ is the pathloss from $\ap_n$ to $\IRS_{l,m}$, $d_{a,l}$ is the distance between the AP and $\IRS_{l}$, and $\omega = \frac{2\pi f}{c}$ is the wave number at frequency $f$. Let us consider the channel between $\IRS_{l}$ and $\user_k$ is $\vec{f}_{l,k}\in \mathbb{C}^{M \times 1}$ and the channel between $\IRS_{l_m}$ and $\user_k$ is given as
\begin{align}
    f_{(l,m),k} = \sqrt{\pl_{(l,m),k}} \, e^{-j \omega d_{l,k}},
\end{align}
where $\pl_{(l,m),k}$ is the pathloss and $d_{l,k}$ is the distance from  $\IRS_{l,m}$ to $\user_k$. Furthermore, the effective cascaded channel vector from the AP to $\user_k$ via $\IRS_{l}$ can be written as
\begin{align}
\vec{h}_{l,k} = \vec{F}_l^{\sf H} \, \vec{\Theta}_l \, \vec{f}_{l,k},
\end{align}
where $\vec{\Theta}_l$ is the IRS reflection matrix, which can be written as
\begin{equation}
    \vec{\Theta}_l = \mathrm{diag}([\kappa_{l,1} e^{j \theta_{l,1}}, \ldots,\kappa_{l,m} e^{j \theta_{l,m}}, \ldots, \kappa_{l,M} e^{j \theta_{l,M}}]),
\end{equation}
where $\kappa_{l,m} \in [0,1]$ and $\theta_{l,m} \in [0, 2\pi)$ represent amplitude coefficients and phase shifts. The channel matrix can be written as 
\begin{equation}
    \vec{H} = [\vec{h}_{1},..,\vec{h}_{k}, \ldots, \vec{h}_{K}]^{\sf H} \in \mathbb{C}^{K \times N}.
\end{equation}
We consider a precoding matrix $\widehat{\vec{W}} \in \mathbb{C}^{N \times K}$ with columns $\widehat{\vec{w}}_k \in \mathbb{C}^{N \times 1}$ denoting unit-norm beam direction for $\user_k$, which can be expressed as $\widehat{\vec{w}}_k = \sqrt{p_k}{\vec{w}}_k,$ where $p_k$ and ${\vec{w}}_k$ are the transmit power scaling factor and the beamforming vector for $\user_k$, respectively. For the given power budget $\vec{P}_{\ap}$, the power constraint can be expressed as
\begin{align}
  \textstyle \sum_{k=1}^K \Vert\widehat{\vec{w}}_k\Vert^2 = \mathsf{Tr} \{ \vec{W}\vec{P} \vec{W}^{\sf H}\} \leq \vec{P}_{\ap},
\end{align}
where $\vec{P} = \mathrm{diag}(p_1,\ldots,p_k,\ldots,p_K)$.
Let $\vec{x} \in \mathbb{C}^{N \times 1}$ be the transmit signal vector from the AP to users, where
$s_k$ is the data symbol intended for $\user_k$, then, the transmitted signal vector from the AP can be written as 
\begin{align} \label{transmit_signal}
	\textstyle\vec{x}=\sum_{{k=1 }}^K\widehat{\vec{w}}_k s_k=\sum_{{k=1  }}^K\sqrt{p_k}{\vec{w}}_k s_k.
\end{align}
The received signals at $\user_k$ through $\IRS_l$ can be expressed as 
\begin{equation}\label{recived_1}
    y_{l,k} = \underbrace{\sqrt{p_{k}}(\mathbf{h}_{l,k})^H {\vec{w}}_k s_k}_{\text{desired signal}} + \underbrace{\sum_{i \neq k}^{K}\sum_{l=1}^{L}\sqrt{p_{j}}(\mathbf{h}_{l,k})^H {\vec{w}}_i s_i}_{\text{interference}} + \underbrace{n_k}_{\text{noise}},
\end{equation}
where $n_k$ denotes the additive white Gaussian noise (AWGN) with zero mean and variance $\sigma^2$. The first term in \eqref{recived_1} represents the signal intended for $\user_k$ reflected by the assigned $\IRS_l$. The second term captures multiuser interference, i.e., the superposition of signals intended for all other users and reflected by all IRSs, which impinges on $\user_{k}$. Furthermore, the signal-to-interference-and-noise ratio (SINR) for $\user_k$ can be written as
\begin{equation}\label{sinr}
    \text{SINR}_{l,k} = \frac{p_{k}|(\mathbf{h}_{l,k})^H {\vec{w}}_k|^2}{\sum_{i \neq k}^{K}\sum_{l=1}^{L}p_j|(\mathbf{h}_{l,k})^H {\vec{w}}_i|^2 + \sigma^2}.
\end{equation}
 Based on the above SINR, the achievable data rate for $\user_k$ can be calculated as
\begin{align}
    R_{l,k} &= \log_2(1 + \text{SINR}_{l,k})\nonumber \\ &=\log_2(1 + \frac{p_{k}|(\mathbf{h}_{l,k})^H {\vec{w}}_k|^2}{\sum_{i \neq k}^{K}\sum_{l=1}^{L}p_j|(\mathbf{h}_{l,k})^H {\vec{w}}_i|^2 + \sigma^2}).
\end{align}
\subsection{Optimization Problem Formulation}\label{OP}
The primary aim of this study is to maximize the network sum rate by optimizing the beamforming matrix and user association. Let $[\vec{\Upsilon}]_{l,k}$ be the user association matrix and $\vec{{W}}$ be the beamforming matrix at the AP. The optimization problem can be stated as follows:
\begin{subequations}\label{eq_opt_prob}
\begin{alignat}{2}
& \vec{P} \text{ : } \underset{ \vec{{W}}\, \vec{\Upsilon}}{\text{ maximize}}
&\quad
&\sum_{k=1}^K R_{l,k} \label{eq_optProb}\\ 
&\quad\text{subject to} 
 &&p_k\geq0, \quad\forall k\in K, \label{eq_constraint1}\\
 &&&
\textstyle  \sum_{k=1}^K p_k\Vert{w}_{k}\Vert^2\leq P_{\ap},\label{eq_constraint2}\\
 &&& \textstyle  \sum_{l=1}^L [\vec{\Upsilon}]_{l,k}=1, \quad\forall k\in K , \label{eq_constraint3}\\
&&&\textstyle \sum_{k=1}^K [\vec{\Upsilon}]_{l,k}=1, \quad\forall l\in L , \label{eq_constraint4}\\
&&&  [\vec{\Upsilon}]_{l,k} \in \{0,1\},\quad \forall k,l,\label{eq_constraint5}
\end{alignat}
\end{subequations}
where~\eqref{eq_constraint1} guarantees nonnegative transmit powers and \eqref{eq_constraint2} limits the total power to the AP power budget $P_{\ap}$. Moreover, constraint~\eqref{eq_constraint3} restricts each IRS to a single user, constraint~\eqref{eq_constraint4} assigns each user to exactly one IRS, and constraint~\eqref{eq_constraint5} enforces binary assignment to complete the one-to-one matching.
\begin{Theorem}\label{T_NP}
    The optimization problem~\eqref{eq_opt_prob} is a MINLP and a non-deterministic polynomial-time hard (NP-hard) problem.
\end{Theorem}
\begin{IEEEproof}
In the context of computational complexity theory, we use three steps to prove that problem~\eqref{eq_opt_prob} is NP-hard. A well-known decision problem $Q$ that has already been shown to be NP-complete is the first one we select. Next, we create a polynomial-time reduction from any instance of $Q$ to an instance of problem~\eqref{eq_opt_prob}. Lastly, we demonstrate that the original problem's objective value is maintained in the modified case. Let us examine the user association problem scenario when the beamforming matrix is provided. The sum-rate maximization problem in~\eqref{eq_opt_prob} can be expressed as
\begin{align}
& \underset{\vec{\Upsilon} }{\text{maximize}}
\quad \sum_{k=1}^K R_{l,k},\nonumber\\&\text{subject to} \quad \eqref{eq_constraint1},\eqref{eq_constraint5},
\end{align}
which is comparable to the well-known NP-hard problem of link scheduling~\cite{mlika2018user}.
The proof of Theorem \ref{T_NP} is thus completed.
\end{IEEEproof}
\section{Proposed Resource Allocation Solutions}\label{proposed}
In this section, we first divide the original optimization problem into two sub-problems, solve each sub-problem, and then combine their solutions to obtain a suboptimal solution to the original optimization problem $\vec{P}$ in \eqref{eq_opt_prob}. By Theorem \ref{T_NP}, the optimization problem $\vec{P}$ is a MINLP and NP-hard problem. Furthermore, the objective of $\vec{P}$ is to maximize the sum rate, which is a nonlinear function, while satisfying multiple constraints. Moreover, the problem structure aligns with a variant multiple-knapsack (VMK) problem. The VMK problem is a generalization of the classical multiple-knapsack problem, in which there are multiple knapsacks and a set of items with associated weights and values. The goal is to select a subset of the items to place in the knapsacks in such a way that the total value of the items is maximized, subject to the constraint that the total weight of the items in each knapsack does not exceed its capacity. Therefore, we use this VMK problem as a tool for decomposing optimization problems. Optimization problem decomposition is the process of breaking down a large, complex optimization problem into smaller, more manageable subproblems that can be solved independently and then combined to obtain a solution to the original problem. This observation motivates a decomposition strategy: (i) beamforming-matrix optimization with fixed association and (ii) user–IRS association with fixed beams. By alternately and optimally solving these two subproblems, we obtain a computationally tractable, suboptimal solution to the original NP-hard formulation.
\subsection{Zero Forcing Beamforming}
 First, we reformulate the optimization problem by fixing the user-IRS association and then optimizing the beamforming matrix under the given association constraints as 
\begin{subequations}\label{P1}
\begin{alignat}{2}
& \vec{P1} \text{ : } \underset{ \vec{{W}}}{\text{ maximize}}
&\quad
&\sum_{k=1}^K R_{l,k} \label{P1.1}\\
&\quad\quad\text{subject to} 
 &&\eqref{eq_constraint1},\eqref{eq_constraint2}.\label{P1.2}
\end{alignat}
\end{subequations}
We propose a ZF beamforming to obtain $\vec{W}$ for the sub-problem $\vec{P1}$ in \eqref{P1}. ZF beamforming is a linear precoding technique widely used in multi-user systems. With ZF, the multiuser interference is eliminated, but it neglects the effect of noise. The fundamental concept of ZF is to use the pseudo-inverse of the channel matrix as the beamforming matrix, which can be expressed as
\begin{align} \label{w_zf}
		& \vec{W}=\vec{H}^\dagger = \vec{H}^H(\vec{H}^H\vec{H})^{-1},
		\end{align}
where $\vec{H}^\dagger$ denotes the pseudo-inverse of channel $\vec{H}$. Furthermore, from the definition of the ZF beamforming, we have
\begin{align} \label{WH_zf}
		& \vec{H}\vec{W}=\vec{H}\vec{H}^H(\vec{H}^H\vec{H})^{-1} =\vec{I}_K.
\end{align}
The interference term vanishes with ZF beamforming, thus, the received signal and SINR at $\user_k$ via $\IRS_l$ can be expressed as, respectively 
\begin{equation}\label{resived_zf}
    y_{l,k}^{ZF}= \underbrace{\sqrt{p_{k}}(\mathbf{h}_{l,k})^H {\vec{w}}_k s_k}_{\text{desired signal}}+ \underbrace{n_k}_{\text{noise}},
\end{equation}
\begin{equation}\label{sinr_zf}
    \text{SINR}_{l,k}^{ZF} = \frac{p_{k}|(\mathbf{h}_{l,k})^H {\vec{w}}_k|^2}{ \sigma^2}.
\end{equation}
\subsection{User-IRS Association}
Then, given the beamforming matrix obtained in subproblem $\vec{P1}$, we formulate the user-IRS association as
\begin{subequations}\label{P2}
\begin{alignat}{2}
& \vec{P2} \text{ : } \underset{ \vec{\Upsilon}}{\text{ maximize}}
&\quad
&\sum_{k=1}^K R_{l,k} \label{P2.1}\\
&\quad\quad\text{subject to} 
 &&\eqref{eq_constraint3}-\eqref{eq_constraint5}.\label{P2.2}
\end{alignat}
\end{subequations}

\begin{algorithm}[!t]
  \caption{Proposed User Association Algorithm for Problem~\eqref{P2}}
  \label{algo:1}
  \DontPrintSemicolon{
  \KwIn{ $K$, $L$, $[\vec{\Lambda}]_{\user_k,\IRS_l}$, $[\vec{\Lambda}]_{\IRS_l,\user_k}$, set of unmatched users $\Pi$, $[\vec{\Upsilon}]_{l,k}$}
 {\bf Initialize }$[\vec{\Lambda}]_{\user_k,\IRS_l} = \emptyset$,
 $[\vec{\Lambda}]_{\IRS_l,\user_k} = \emptyset$, $[\vec{\Upsilon}]_{l,k} = \emptyset$;\;
\underline{\textbf{Phase 1: Preference Matrix:}}\\
\For {$k\in K$}{
$R_{l,k}, \quad \forall\hspace{2mm} l \in L$ and store in $[\vec{R}]_{l,k}$\;
}
$[\vec{R}]_{k,l}= ([\vec{R}]_{l,k})^T$\;
[$[\vec{R}]_{l,k}$,$[\vec{\Lambda}]_{\IRS_l,\user_k}$] = sort ($[\vec{R}]_{l,k}$,2,descend)\;
[$[\vec{R}]_{k,l}$,$[\vec{\Lambda}]_{\user_k,\IRS_l}$] = sort ($[\vec{R}]_{k,l}$,2,descend)\;
\underline{\textbf{Phase 2: User-IRS Association:}}\\
\While {(either $\Pi$ $\neq \emptyset$ or users not rejected by all IRSs)}{
\For {$\user_{k'}\in \Pi$}{
make a proposal to the IRS with the highest preference in $[\vec{\Lambda}]_{\user_k,\IRS_l}$\;
set element of $[\vec{\Upsilon}]_{l,k} = 1$\;
\For {$\IRS_l\in L$}{
{\bf if} {($\IRS_l \notin [\vec{\Upsilon}]_{l,k}$ )}\;
{
$[\vec{\Upsilon}]_{l,k} \gets [\vec{\Upsilon}]_{l,k} \cup (\user_{k'},\IRS_l) $\;
}
{\bf else if} {($R_{l,k'}>R_{l,k}$)}\;{
$[\vec{\Upsilon}]_{l,k} \gets [\vec{Psi}]_{l,k} \cup (\user_{k'},\IRS_l)$ \;
{\bf else}\;
$[\vec{\Upsilon}]_{l,k} \gets [\vec{\Upsilon}]_{l,k} $ \;
}
}
  }
{\bf Output: }Association matrix $\vec{\Upsilon}^\star$
}}
\end{algorithm}
We deploy multiple TIRSs and a single AIRS to assist the communication links. The user-IRS association $\vec{P2}$ is modeled as a one-to-one matching problem in which each user can be matched to at most one IRS in a slot. Algorithm \ref{algo:1} implements a user-proposing deferred-acceptance procedure. Initially, all users are unmatched and maintain preference lists of IRSs ordered by their expected data rates as described in phase 2 of Algorithm~\ref{algo:1}. In each round, every unmatched user proposes to the most preferred IRS. Each IRS tentatively accepts the proposal that yields the highest data rate among its current suitor and any existing tentative match and rejects the others. Rejected users remove that IRS from their lists and may propose to their next choice in subsequent rounds. The process repeats until no further proposals are possible, i.e., all users are matched or have exhausted their lists, yielding a stable matching under the given preferences as shown in phase 2 of Algorithm~\ref{algo:1}.
\begin{Theorem}\label{T2}
Consider Algorithm \ref{algo:1} with $K$ users and $L$ IRSs. Counting one iteration as a single proposal event, the proposed algorithm terminates after at most $K\times L$ iterations; hence, it runs in polynomial time $O(K\times L)$.
\end{Theorem}
\begin{IEEEproof}
Each unmatched user proposes only to IRSs it has not yet tried. No user ever proposes to the same IRS twice. Therefore, the total number of distinct proposals across all users is at most $K\times L$. Upon each proposal, the recipient IRS either is free and tentatively accepts the user or is engaged and compares the new proposer with its current tentatively matched user, keeping the preferred one and rejecting the other, as shown in phase 2 of Algorithm~\ref{algo:1}. In either case, the set of untried IRSs for the proposing user strictly shrinks. Because the number of possible proposals is finite ($\le K\times L$), the process must terminate after at most $K\times L$ iterations. At this point, either all users are matched or all proposal lists are exhausted. Hence, the algorithm runs in $O(K\times L)$ time.
\end{IEEEproof}
\section{Simulation Results and Discussion}\label{sim}
We consider a network with multiple users and a single multi-antenna AP deployed over the network area. Several IRSs are installed, and each IRS comprises $100\times100$ reflecting elements with half-wavelength element spacing ($M_y=M_z= \lambda/2$). The simulation parameters are summarized in Table~\ref{tab:sim}. All algorithms are implemented in MATLAB. Unless otherwise stated, each plotted point is averaged over $10^6$ independent channel realizations. We evaluate the proposed method against three benchmarks. First, the ES baseline enumerates all feasible one-to-one user–IRS assignments and selects the assignment that maximizes the network sum rate. Second, the GS-based allocation allows each user to select the RS with the highest data rate. However, IRSs receiving multiple proposals are selected randomly. Third, the RS baseline randomly selects the user-IRS association matrix. 
\begin{table}[!t]
\centering
\renewcommand{\arraystretch}{1}
\caption{Simulation parameters}
\label{tab:sim}
\begin{tabular}{l l}
\hline
\textbf{Parameters}                                  & \textbf{Values} \\ \hline
Carrier frequency ($f_c$)  {[}GHz{]}                   & $15$              \\
Number of antennas at AP& $256$\\
AP power budget {[}dBm{]}                    & $43.2$            \\ 
Power density of noise {[}dBm/Hz{]} & $-174$           \\ 
Number of reflectors,                               & $100\times100$            \\ 
Channel bandwidth {[}MHz{]}                     & $400$ \cite{karaman2025demand}           \\ 

Noise figure {[}dB{]}                           & $10$ 
\\ 
Side length of reflectors ${[}m{]}$  & ${\lambda}/{2}$\\ \hline 
\end{tabular}
\end{table}
Fig.~\ref{sum_rate_power} shows the sum rate versus transmit power. As expected, all schemes exhibit a monotonically increasing trend. The proposed JBUA achieves performance within $2\%$ of the ES benchmark while consistently outperforming the GS and RS baselines across the entire power range, yielding approximately $14\%$ and $44\%$ higher sum rates than GS and RS, respectively.
\begin{figure}[!htp]
     \centering
\includegraphics[width=\linewidth]{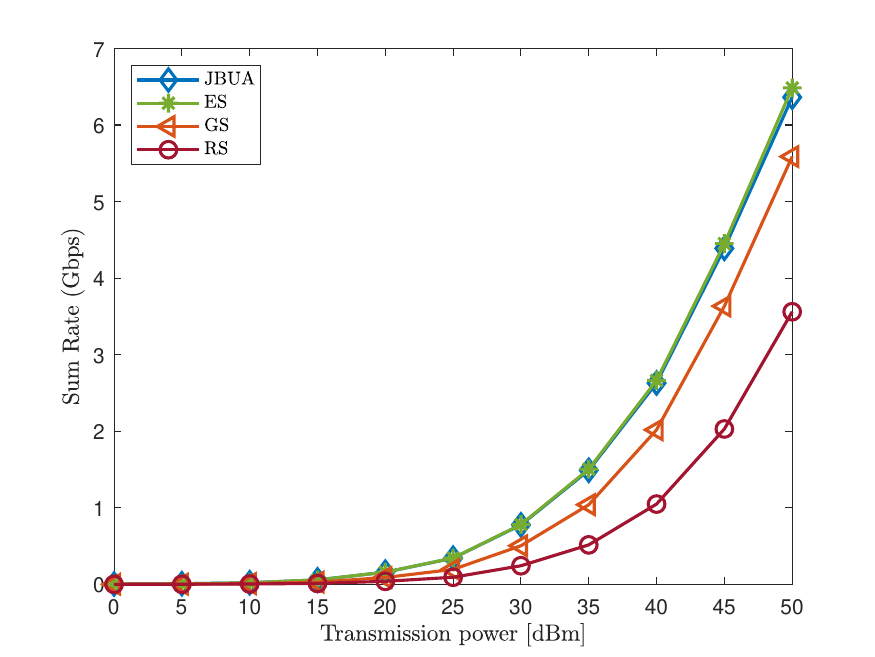}
    \caption{Sum rate versus AP power budget.
     }
     \label{sum_rate_power}   
\end{figure}
\begin{figure}[!htp]
     \centering
\includegraphics[width=\linewidth]{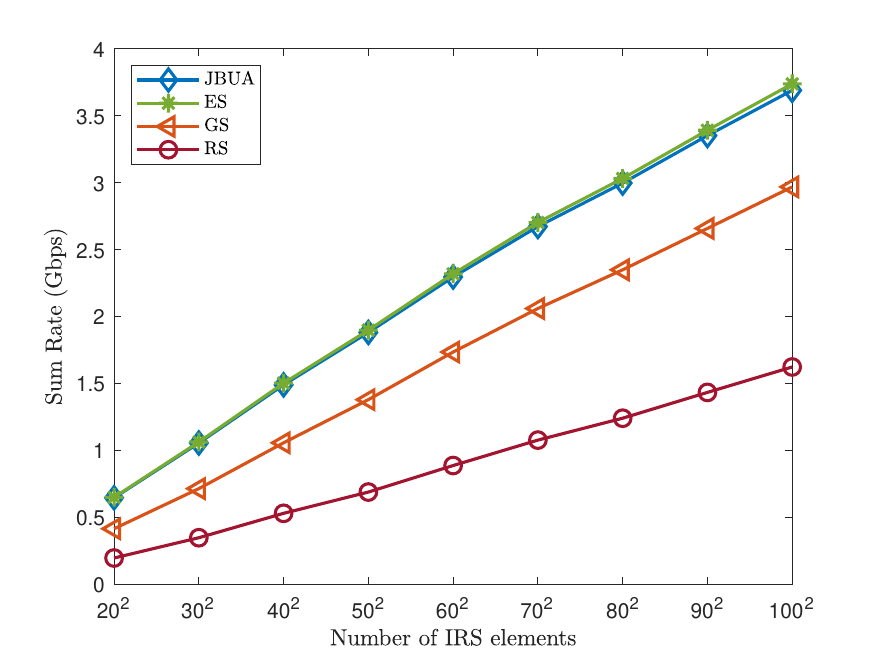}
    \caption{Sum rate versus the number of IRS elements.
     }
     \label{sum_rate_reflectors}   
\end{figure}

Fig.~\ref{sum_rate_reflectors} shows the sum rate versus the number of reflectors. As expected, all schemes improve monotonically with increasing reflector count. Across the entire range, the proposed JBUA closely approaches the ES benchmark while consistently outperforming the GS and RS baselines. More specifically, for IRSs with $100\times100$ elements, JBUA achieves approximately $22\%$ and $56\%$ higher sum rates than GS and RS, respectively. Fig.~\ref{sum_rate_network} shows the network sum rate versus network size. The sum rate decreases as the coverage area grows due to larger AP-user distances and the resulting path loss. Across all sizes, JBUA remains close to ES and consistently outperforms GS and RS
\begin{figure}[!htp]
     \centering
\includegraphics[width=\linewidth]{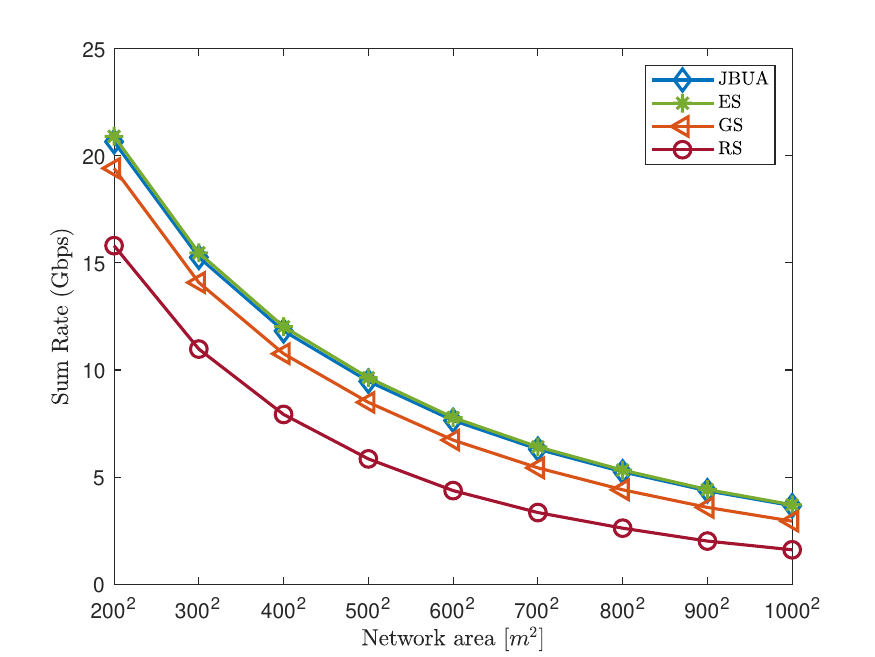}
    \caption{Sum rate versus network area.
     }
     \label{sum_rate_network}   
\end{figure}
\section{Conclusions}\label{conc}
This paper studied robust resource allocation for futuristic 6G networks operating in the FR3. We proposed a multi-layer IRS-assisted communication framework that mitigates severe LoS blockages via a two-tier structure comprising TIRSs and AIRS. To maximize network performance, we formulated the JBUA problem and developed a stable-matching-based algorithm that selects user-IRS pairs according to achievable data rates. Numerical results demonstrate that the proposed approach consistently improves the downlink sum rate and enhances IRS utilization compared with baseline schemes. In particular, the JBUA solution approaches the performance of exhaustive search while retaining substantially lower complexity, thereby validating the effectiveness of the proposed multi-layer IRS architecture for reliable 6G deployments in FR3. This research can be expanded by using a LEO satellite segment to improve the adaptability, coverage continuity, and scalability of the hierarchical IRS design for future 6G deployments. In addition, we will investigate hybrid beamforming at the AP and across IRSs during user mobility, accounting for tracking, Doppler, and handover to better mimic real-world operating conditions.
\bibliographystyle{IEEEtran}
\balance
\bibliography{References}

\end{document}